\documentclass[copyright,creativecommons]{eptcs}
\providecommand{\event}{}
\providecommand{\volume}{??}
\providecommand{\anno}{20??}
\providecommand{\firstpage}{17}
\providecommand{\eid}{??}
\usepackage{breakurl}

\usepackage[usenames]{color} 
\usepackage{amsmath}
\usepackage{randtext}

\usepackage[shorthand,reserved,inference]{semantic} 
\usepackage{latexsym}
\usepackage{amssymb}
\usepackage{amsfonts}
\usepackage{stmaryrd}

\usepackage{amsthm}% added this for proof environment

\usepackage{listings} % program listings
\usepackage{graphicx} % color used in listings

\theoremstyle{plain}
\newtheorem{theorem}{Theorem}[section] % chapter]
\newtheorem{definition}[theorem]{Definition}

\theoremstyle{definition}
\newtheorem{example}[theorem]{Example}

\theoremstyle{remark}

\mathlig{<<}{\langle}
\mathlig{>>}{\rangle}

\newcommand{\labarrow}[1]{\stackrel{#1}
                                   {\rightarrow}}

\newcommand{\opene} {|[} % {\pmb{[}}  % boldface symbol (ugly?)
\newcommand{\closee}{|]} % {\pmb{]}} % boldface symbol

\newcommand{\ersys}[1]{\opene #1}    % {a!BE?#1} 
\newcommand{\erusr}[2]{\opene #1,#2} % {a?BE?#1!#2}
\newcommand{\ercmp}[2]{\opene #1,#2} % {BE,#1,#2}

\newcommand{\outeq}{\stackrel{o}{=}}

\newcommand{\ube}[2]{\labarrow{\erusr{#1}{#2}}} %user begin erasure, 
\newcommand{\sbe}[1]{\labarrow{\ersys{#1}}}
\newcommand{\cbe}[2]{ \labarrow{\ercmp{#1}{#2}}}
\newcommand{\enErase}[1]{\labarrow{#1\closee}} 
\reservestyle{\keyword}{\mathbf}
\keyword{in,input,erased,erase,on,output,to,from,read,write,if,fi,while,done,true,false,then,else,skip,out}

\lstdefinelanguage{WHILE}{
  keywords={output, input, erased, in, mod, div, to, int, while, if, then, else, for, on, do, skip },
}

\reservestyle{\infixkeyword}{\mathrel\mathtt}
\infixkeyword{:=,;,==}

\usepackage{changebar}

\title{A User Model for Information Erasure}
\author{ Filippo Del Tedesco
\institute{Chalmers University of Technology\\ Gothenburg, Sweden}
\email{\randomize{tedesco@chalmers.se}}
\and
David Sands
\institute{Chalmers University of Technology\\ Gothenburg, Sweden}
\email{\randomize{dave@chalmers.se}}
}
\def\titlerunning{A User Model for Information Erasure}
\def\authorrunning{Del Tedesco \& Sands}
\begin{document}
\maketitle

\begin{abstract}
%This is the final sentence in the abstract.
%
  Hunt and Sands (ESOP'08) studied a notion of \emph{information
    erasure} for systems which receive secrets intended for
  limited-time use. Erasure demands that once a secret has fulfilled
  its purpose the subsequent behaviour of the system should reveal
  no information about the erased data. In this paper we address a
  shortcoming in that work: for erasure to be possible the user who
  provides data must also play his part, but previously that role was
  only specified informally. Here we provide a formal model of the
  user and a collection of requirements called \emph{erasure
    friendliness}. We prove that an erasure-friendly user can be
  composed with an erasing system (in the sense of Hunt and Sands) to
  obtain a combined system which is \emph{jointly erasing} in an
  appropriate sense. In doing so we identify stronger requirements on the user than those informally described in the previous work. 
\end{abstract}

%%%%%%%%%%%%%
%%THIS IS A SECTION%
%%%%%%%%%%%%%
\section{Introduction}\label{sec:intro}

The requirement that data is used but not retained is commonplace. As
an everyday example consider the credit card details provided by a
user to a payment system. The expectation is that card details will be
used to authorize payment, but will not be retained by the system once
the transaction is complete.

The study of erasure policies from a language-based security
perspective was initiated by Chong and Myers
\cite{er1}. Hunt and Sands \cite{Hunt:Sands:ESOP08}
argue that to give a satisfactory account of erasure, one needs to
consider \emph{interactive} systems: the card details are used in the
interaction between the customer, the payment system and the bank, and
\emph{then} erased; without the interaction, the card details could be
dispensed with altogether and erasure would be unnecessary. They
present an information-flow based definition of erasure for sequential
programs interacting with users through channels, and a type system
for guaranteeing erasure properties.

This paper deals with the model of erasure described by Hunt and
Sands, and addresses a shortcoming of that work: despite the emphasis
on an interactive view of erasure, there is no explicit model of the
users of the program. As a consequence certain requirements about the
users are described but not formalized. More specifically, users who
supply data for erasure are expected to fulfill certain
obligations. These obligations are intended to complement the erasure
property fulfilled by the program so that together users and the system
truly achieve their data erasure goals. Since users and their
obligations are not modeled, the previous work is unable to prove
this. The main result of this paper is to formalize the obligations of
what we will call an \emph{erasure friendly} user, and to show that an
erasure friendly user when composed with an erasing system satisfies \emph{composite erasure}, which defines  a joint responsibility between user and system in the erasure mechanism.
% is \emph{jointly erasing}. %\draftnote{Need to come up with a good name for the new erasure property of the combined system}.  
As a result of the formalization of the user we discovered additional obligations
required on erasure friendly users which were not described in
\cite{Hunt:Sands:ESOP08}.

In the remainder of this introduction we outline the basic ideas in
the definition of an erasing program from \cite{Hunt:Sands:ESOP08},
and the informal obligations identified for the user of such a
program. Finally we outline remainder of the paper.

\paragraph{Erasure: the basic idea}
%\draftnote{This maybe belongs in the section about erasure?}  
For the purposes of this paper we will work with a simplification of the
erasure notion from \cite{Hunt:Sands:ESOP08}. In that work (following
\cite{er1}) data is labelled using a multilevel
security lattice and erasure is from one level to some higher
level. 
%Here we will ignore the multilevel security issues and assume data from one specific user (level).  

Here we will ignore the multilevel security issues and assume data from one specific user (level), which are requested and erased completely by the system through the following programming construct $\<input>~ x~
\<erased>~\<in>~C$.

%SWITS
%Given this simplification, the programming construct, which we can write as $\<input>~ x~
%\<erased>~\<in>~C$.
%\draftnote{This sentence looks strange to me. Is it ok? I would say "}
Operationally this behaves just as an input statement from the user
which writes into a variable $x$, after which the command $C$ is
executed. But the intention is that this \emph{specifies an erasure
  policy}, namely that the data input will not be used beyond the
command $C$.  If this property holds for all runs and all erasure
blocks then the program is said to be \emph{input erasing}.  To make
this notion precise a noninterference-style definition is used:
suppose that a given computation reaches a statement of the form
$\<input>~ x~ \<erased>~\<in>~C$, and the user provides an input value
$u$ for $x$, the computation reaches the end of the command $C$ and
after this point the program has some observable behavior $t$. The
erasure requirement is that if the user had instead provided some
different value $v$, then the program would still be able to reach
the end of the command $C$ and produce the same observable behavior
$t$. This ensures that after the end of the erasure block the observer
can learn nothing about the user's input by analyzing the system behavior. Then the definition has also to prevent the secret provider being exploited as an involuntary storage for the secret, requiring that its behavior after $C$ has to be the same independently of the value it has sent. Considering a simple model for user (for example a list of values, like in \cite{Hunt:Sands:ESOP08}), this is tantamount requiring that the number of inputs provided by the user to the system has to be the same for all possible $C$ executions.
% \footnote{Note that this is a behavioral notion of erasure; to make
%   it more ``physical'' we simply require an observer who makes lower-level
%   observations}

As an example consider the pseudocode in Figure~\ref{fig:prog} which
depicts a credit card transaction with input from the user. The
intention is that the credit card is erased after the completion of
the transaction, and this is reflected in the reassignment of the credit card variable on line~8. 
In fact the credit card is \emph{not} erased correctly here.  When the server
goes down information about the credit card is retained
via the variable \texttt{payment}, which is outputted to the log file in the last statement of the program. To
make the program input erasing one must additionally overwrite
\texttt{payment} at the end of the while loop.

Note that there are two kinds of inputs from the user: those subject
to erasure, and those (like the input of the shipping address) which
are not. Additionally, outputs during the erasure block may contain
information that is the subject of the erasure, as for example the
echoing of the order information on line~4.

\begin{figure}[htb]
 \begin{lstlisting}
 while serverUp {
     input creditCard erased in { 
         input shippingAddress; 
         output (creditCard++shippingAddress) to user;
         payment := process(creditCard); 
         output payment to bank 
         custDatabase := (custDatabase ++ shippingAddress);
         creditCard := 0 
     } 
}; 
output all variables to logfile;
  \end{lstlisting}
\caption{Example program}
\label{fig:prog}
\end{figure}

\paragraph{User Obligations}

Certain user obligations and assumptions are implicitly built into the
notion of input erasure described above. One is that we clearly cannot
expect erasure to ``work'' for an arbitrary user. For example, if a
user adopts the strategy of always appending his credit card number to
his shipping address then since the shipping address is included in
the customer database the system would inadvertently save the credit
card number and the combined system would not be erasing as intended.

A second example from \cite{Hunt:Sands:ESOP08} reveals more
assumptions about the user: suppose that, before the credit card is erased,
the program provides the user with a special offer
code with the promise ``present this code when
you next shop with us for a 10\% discount''. What if this code is
simply an encryption of the credit card number?  Although it may be reasonable to assume that the user does not exhibit the deliberately bad strategy described above, how do we ensure that the user does not re-input this code after the transaction (since re-inputing the code will enable the program to recover the credit card number)?
Hunt and Sands argue that, in contrast with noninterference, 
it is \emph{not} reasonable for the user to know the semantics of the system and be able to make perfect deductions about what is safe. It was proposed that
the user strategy:  
\begin{itemize}
\item assumes that the user knows that certain data is scheduled for
  erasure, and they are notified when the erasure is complete, and
\item assumes that the user treats any output from the program 
as potentially tainted with
data currently scheduled for erasure.
\end{itemize}
The question we answer in this paper is whether this user strategy is indeed 
sufficient to ensure the desired erasure property. 

\paragraph{Outline} %\draftnote{This is only a provisional outline.}
The approach we take to answer these questions is as follows. The
first step (Section~\ref{sec:erasure-system}) is to recall the notion of input
erasure for a program (what we will henceforth call the \emph{system}). We deviate from the original formulation by representing this in a
syntax-independent form. To do this we make the beginning and end
of an erasure session externally visible as communication events which
declare, respectively, that the next input from the user will later be erased, and  that the input has been erased.  We refer to the input erasure property of a system $S$ as $E(S)$.

The second step (Section~\ref{sec:user}) is to define the interface
and structure of a user, $U$, who will provide the data for erasure. 
We then specify both the communication mechanism between $U$ and $S$ as well as a number of semantic constraints on user behavior which define when $U$ is what we call \emph{erasure friendly}, $EF(U)$. 

Finally (Section~\ref{sec:thm}) we define the desired erasure property, \emph{composite erasure}, for the combined system $EC(U | S)$. We are then able to prove the main theorem, namely that if we combine an input-erasing system with an erasure-friendly user then we obtain a combined system which is erasing: 
\[
E(S) ~\&~ EF(U) \Rightarrow EC(U | S)
\]
Related works are discussed in Section~\ref{sec:related}, and conclusions and further work are outlined in Section~\ref{sec:conclusions}.

% \begin{itemize}
% \item models user as a stream (justified from a NI perspective but not clear what this means from a user perspective)  
% \item informal description of requirements of a user
% \end{itemize}

% Goals of this work
% \begin{itemize}
% \item Model the user providing data to determine whether the user 
% \end{itemize}

% Where to put comparison between erasure and 

%%%%%%%%%%%%%
%%THIS IS A SECTION%
%%%%%%%%%%%%%
\section{Systems and Abstract Input Erasure}\label{sec:erasure-system}

In the previous approach systems were defined through a deterministic
imperative language, which was equipped with input-output primitives.  As we have already seen in
the Example~\ref{fig:prog}, for erasure the most important one is the
block structured input command, where the value received through the
input operation must be erased at the end of the block.
In this section we introduce our system model and corresponding
erasure definition. In some respects it is more general and abstract
and in other respects it is simplified when compared to the earlier work. 

\begin{itemize}\item 
  The generalization comes from handling a model which is independent
  of a particular programming language syntax. This is possible
  because we are focussing on the interaction between user and system,
  rather than the \emph{verification} of the erasure property of
  systems.  What we retain from the earlier model is the assumption of
  determinism, and a notion of \emph{well-formedness}
  which includes an abstract
  counterpart to the use of a block-structured erasure construct.
\item
  The simplification here is that (i) we focus on erasure of data from
  a single user, and (ii) we assume that the information is totally
  erased, rather than the more general case of being erased to a
  higher security level. Because of these simplifications we no longer
  need to model a multilevel security lattice. Instead we simply
  assume that there is one other user of the system whose security
  level is different from the user supplying the data to be erased.
\end{itemize}

% In this work it is no longer necessary to have such a low level perspective on the system, because here we are focussing on the main aspect of the interaction between the system and its user, namely their \emph{observable behavior}, from which we can understand if erasure was performed or not. Hence we propose a more abstract description for the system, which should be able to capture the essential elements of erasure with a clear and simple notation.

% Also here a system is an agent which is able to interact with its environment, as well as perform internal computations based on its internal state and the values received. In order to keep the presentation as simple as possible, we assume that each system is able to communicate in a bidirectional way on a single channel $a$ (which will be shared with the user), while it is able to use channel $H$ and $L$ just for output operations. $H$ and $L$ are symbolic representatives for channels $c \sqsupseteq a$ and $c \not \sqsupseteq a$ respectively, and they represent a simple way to check the compliance of $S$ behavior to erasure policy: once the secret is erased, the stream of data through those channels has to be independent with respect to the secret value received.

Since we will model erasure properties in terms of inputs and
outputs we are not interested in the representation of system's internal state. Thus we are able to describe it through a labelled transition system
$S=(\mathcal{S},\mathcal{L},\mathcal{T})$, in which the components have the following meaning:
\begin{itemize}
\item $\mathcal{S}=\{s_i: i$ is an index$\}$ is the set of (distinct) states, in which $s_0$ is the initial one. The counterpart of this element in the programming language approach was the entire code for the program, together with the empty memory. 
%\draftnote{Probably $\mathcal{S}$ is not necessary, but I find quite comfortable to have it}
\item $\mathcal{L}=\{a!v,a?v,a!BE,a!EE,b!v\}$ is the set of labels for $S$: the first four are used for the communication with the user $U$ through channel $a$. The other represents an output event on a different channel $b$.
While $v$ represents an ordinary value received or sent during an I/O interaction (which is a member of the domain of values $\mathbb{V}$), $BE$ and $EE$ are considered reserved values, which mark the beginning and end of an erasure ``block''.
\item  $\mathcal{T}\subseteq \mathcal{S} \times \mathcal{L} \times \mathcal{S}$ is the set of transactions such that if a pair of values $v,s_v$ exists such that $(s, a?v, s_v) \in \mathcal{T}$, then for each $w \in \mathbb{V}$ a state $s_w$ exists such that $(s, a?w, s_w) \in \mathcal{T}$. Basically we require that $S$ is an \emph{input enabled} system, ready to accept each possible value offered on channel $a$.
% \draftnote{I tried to include here the "input enabled" property.}
\end{itemize}
We need to impose some further restrictions on the behavior of the system, the first of which is determinism; we take the definition from \cite{str2}.  
% This definition of the system is too liberal and does not capture the old class of systems: for example we need to impose on $S$ a deterministic behavior (this definition came from \cite{str2}).
\begin{definition}[Deterministic System]
A system $S$ is deterministic if:
\begin{itemize}
\item whenever $s \stackrel{l_1}{\rightarrow} s_1$ and $s \stackrel{l_2}{\rightarrow} s_2$ and $l_1 \not = l_2$ then $l_1=a?v_1$ and $l_2=a?v_2$ for some values $v_1$, $v_2$. %\draftnote{$v_1 \neq v_2$? - redundant}
\item If $s \stackrel{l}{\rightarrow} s_1$ and $s \stackrel{l}{\rightarrow} s_2$ then $s_1=s_2$.
\end{itemize}
\end{definition}
%%%%%%%%%%%%%%%%%%%%%%%%%%%%%%%%%%%%%%%%%%%%%%%%%%%%%%%%%%%%%%%%%%%%

% Say something about the use of determinism?

Now  we can proceed towards the erasure part of system specification. Since this notion is closely related to the system behavior, it is convenient to use the notion of \emph{trace} which records a sequence of possible interactions from a given state. 

%, which basically represents the evolution of the system from one of its state to another, possibly performing several distinct steps.

\begin{definition}[Traces of $S$]
Let  $s_i \stackrel{l_0}{\rightarrow} s_{i+1} \stackrel{l_1}{\rightarrow}  \cdots \stackrel{l_{n-1}}{\rightarrow} s_{i+n}$ be a sequence of $n\geq 0$ adjacent edges in $S$. The concatenation of labels $l_0l_1\cdots l_{n-1}$ is a trace for $s_i$, and can be represented as $s_i\stackrel{l_0l_1\cdots l_{n-1}}{\twoheadrightarrow} s_{i+n}$.  If we are not interested in the final state we can write it as $s_i\stackrel{l_0l_1\cdots l_{n-1}}{\twoheadrightarrow}$, while if we are not interested in the trace per se, but just in saying that $s_{i+n}$ is reachable from $s_i$ we write $s_i \twoheadrightarrow s_{i+n}$. A trace $t_1$ is a prefix of a trace $t_2$ ($t_1 \preceq t_2$) iff  a (possibly empty) trace $t_3$ exists such that $t_1t_3=t_2$. The set of all traces of $s_i$ is denoted by $T(s_i)$, and it always contains the empty trace $\epsilon$ which corresponds to a sequence of zero edges from $s_i$. The traces of $S$ are denoted by $T(s_0)$ or, equivalently, by $T(S)$.
\end{definition}

The next constraint we impose on the system concerns the structure of
the erasure operation. As noted in \cite{Hunt:Sands:ESOP08}, the user
needs to be aware of not only when a value supplied is subject to
later erasure, but also the point at which the erasure of an input is
intended to be complete.  For this purpose we have the outputs $a!BE$
and $a!EE$, by which $S$ communicates the beginning and the end of the
erasure session to the user. The protocol we assume is that an
input from the user which is subject to erasure is % followed by
preceded by a $BE$ message, as in $s_1 \stackrel{a!BE}{\rightarrow}
s_2 \stackrel{a?v}{\rightarrow} s_3$.  A communication $s
\stackrel{a!EE}{\rightarrow} s'$ will be the signal to the user that
the erasure is now complete. But to which input does the signal refer?
Here we use an analog of the block-structure: we assume that the $BE$
and $EE$ are well bracketed so that each $EE$ message uniquely
identifies a corresponding $BE$ and input.

% Unfortunately the system is still too liberal: actually the next aspect we need to consider is how to involve the user in the erasure process. As we have already mentioned in the previous section, we need a mechanism, based on the reserved labels $a!BE$ and $a!EE$, to let $S$ communicate the beginning and the end of the erasure session to its user.  Our choice is to distinguish inputs of non-erasing values (represented by $a?v$) from erasing inputs by joining the input operation with the begin erasure label, as in $s_1 \stackrel{a!BE}{\rightarrow} s_2  \stackrel{a?v}{\rightarrow} s_3$, which will be closed by an appropriate $s \stackrel{a!EE}{\rightarrow} s'$ for a certain $s$. The evolution from $s_3$ to $s$ is what we call an \emph{erasure session}. \draftnote{Although the introduction contains those element in the other way around, it seems to me unclear then how to pick a value from $\delta$ without knowing that a secret has been required}

To refine the definition of $S$ according to those intuitions we proceed in two steps: first we propose a constraint on system structure related to the use of $BE$ and $EE$ actions, then we state the erasure definition.

\begin{definition}[Well-Formed System]

Let $s \sbe{v} s_1$ abbreviate $s \stackrel{a!BE}{\rightarrow} s' \stackrel{a?v}{\rightarrow} s_1$.
A system $S$ is well formed iff for each $t \in T(S)$: 
\begin{itemize}
\item if $t=t'(a!BE)$, and $s_0 \stackrel{t'}{\twoheadrightarrow} s$, then $s \sbe{v}$ for each possible value $v \in \mathbb{V}$;%
%
%\draftnote{Could require this for all inputs -- looks like a standard property for value passing systems. Goto the previous draftnote}
%
\item in each prefix $t'$ of $t$, the number of $\ersys{v}$ labels is always greater or equal than the number of $a!EE$ labels;
\item if $s_0 \stackrel{t}{\twoheadrightarrow} s \not \rightarrow$ then the number of $\ersys{v}$ in $t$ is equal to the number of $a!EE$ in $t$.
\end{itemize}
\end{definition}
% Intuitively, the above definition states that a system, to be well-formed, has to use the erasure related commands in an appropriate way, which means every time an erasure input is performed, there will be an associated closure. Since in a single trace there could be more then one erasure related input, we assume a \emph{well nested structure} among them, which means that the first erasing input operation in $t$ will be closed by the last $a!EE$ in $t$ and so on. 

%So, for a well-formed system, whenever $s_0\twoheadrightarrow s_1\sbe{v}s_2$ there will be a trace $u$ and a state $s_3$ such that $s_2 \stackrel{u}{\twoheadrightarrow}s_3 \stackrel{a!EE}{\rightarrow}$ and ($\ersys{v}$, $a!EE$) will be the border of erasure for the action $a?v$. We will refer to $u$ as the \emph{erasure session}. Henceforth $s_2 \stackrel{u}{\twoheadrightarrow}s_3 \stackrel{a!EE}{\rightarrow}$  will be written as  $s_2 \enErase{u}$ or $s_2 \enErase{}$ if we are not interested in the trace $u$.

%\DTNoteB{
%So whenever $s_0\twoheadrightarrow s_1\sbe{v}s_2$ and the system is able to go on with a trace $u$ reaching a state $s_3$ such that $s_2 \stackrel{u}{\twoheadrightarrow}s_3 \stackrel{a!EE}{\rightarrow}$,  where $a!EE$ closes the correspondent $a!BE$, we will say that ($\ersys{v}$, $a!EE$) is the border of erasure for the action $a?v$ and $u$ is the \emph{erasure session}.} Henceforth $s_2 \stackrel{u}{\twoheadrightarrow}s_3 \stackrel{a!EE}{\rightarrow}$  will be written as  $s_2 \enErase{u}$ or $s_2 \enErase{}$ if we are not interested in the trace $u$.

%\DTNoteB{
Let us consider an execution of $S$ such that $s_0\twoheadrightarrow s_1\sbe{v}s_2$. If a state $s_3$ exists such that $s_2 \stackrel{u}{\twoheadrightarrow}s_3 \stackrel{a!EE}{\rightarrow}$ and $a!EE$ corresponds, according to well bracketing, to the $\sbe{v}$ action, then we write $s_2 \enErase{u}$. Note that $u$ may be infinite, in that case $\sbe{v}$ does not have a matching $a!EE$.
%}

Now the last aspect we need to consider to translate the previous
description for systems into the new abstract setting is related to the
input erasure definition. In the original definition the user is represented simply by a stream of input
values. We represent this in our more abstract setting by specializing $S$ to a particular input stream. (Later we will show that this representation of a user as a stream is indeed sufficient in a certain sense).
\begin{definition}
  Let $I$ be a stream %, $I'$ range over streams of values. 
  and, for any $n > 0$, let $I(n)$ 
  denote the $n$-th value of the stream.  For any stream $I$ 
  let $S(I)$ denote the refinement of $S$ in which all
  input nondeterminism has been resolved as follows: $t \in T(S(I))$ if and only if
  $t \in T(S)$, and the for all $n > 0$, the $n$-th input label occurring in $t$ has value $I(n)$.
\end{definition}
%SWITS
%\draftnote{Give small example?}
Note that the transition system for $S(I)$ is just
a single trace (potentially infinite) of $S$.

%\draftnote{Fuzzy description. Formal definition of $S(I)$?}
% First of all we need to retrieve the old I/O environment: in the previous approach the system was associated to a set of streams, one for each channel, which basically are (possibly infinite) list of values fixed before the computation starts (each value provided by a stream depends just on the number of values requested previously). Since $S$ was deterministic, defining a family of streams for $S$ was the same of resolving each possible choice in $S$, actually obtaining one of its possible executions. 

% In our new framework a stream-like environment is meaningful as well: although the family contains exactly one element (for the channel $a$, which is the only input source for $S$), associating a stream $I$ to $S$ is exactly like considering one \emph{instance} $S(I) \in T(S)$. 

% According to this perspective, the old \emph{input erasure} property can be easily restated in this way.

\begin{definition}[Abstract Input Erasure $E(S)$]
Let $S$ be a deterministic, well-formed system. Let $i(t)$ be an operator over traces which counts the number of input labels (of the form $a?v$) contained in $t$. Consider an arbitrary stream $I$ and the associated refinement $S(I)$, such that $s_0 \stackrel{t}{\twoheadrightarrow} s_1 \sbe{v} s_2 \enErase{u} s_3 \stackrel{z}{\twoheadrightarrow}$. 
Suppose that the input $v$ is the $n$-th input of the trace. 
Then $S$ is input erasing if for any $I'$ which differs from $I$ only at position $n$, we have that 
% $I'=_{\neg i(t)+1} I$, $I'(i(t)+1)=w$, it holds that 
$S(I')$ has the trace $s_0\stackrel{t}{\twoheadrightarrow} s_1 \sbe{w}s_2' \enErase{u'} s_3' \stackrel{z}{\twoheadrightarrow}$ and $i(u)=i(u')$.
\end{definition}
%Here $I'=_{\neg i(t)+1} I$ means that streams $I'$ and $I$ are the same but in position $i(t)+1$. 
As we can see, the definition says that after erasure, the subsequent behavior of the system is independent of the value supplied at the beginning of the erasure session. The only constraint on the behaviour during the erasure session is that the number of inputs does not depend on the value which is subject to
erasure (for more details on the motivation for the latter point see \cite{Hunt:Sands:ESOP08}).
% an Abstract Input Erasing system is able, independently from the subject of the erasure session $v$, 
% to close the erasure session in a way that its behavior is the same (although states $s_3$ and $s_3'$ could be internally different, they behave in the same way) and the surrounding environment (the stream which provides value) does not provide a way to infer the value of the erased secret. We address the reader interested in more explanations and examples to \cite{Hunt:Sands:ESOP08}.

%%%%%%%%%%%%%
%%THIS IS A SECTION%
%%%%%%%%%%%%%
\section{Erasure Friendly Users}\label{sec:user}

In this section we define the general user model, and the requirements
that we need to impose on such a user, \emph{erasure friendliness}, to
ensure that its composition with an input-erasing system achieves composite erasure.
%\ldots \draftnote{Complete with property name. F: joint erasure seems to me clear and reasonable. Why do not we use it?}

Since a reasonable scenario for erasure should involve either human or
computer agents, the model of the user of an input erasing
system % has to exploit their common aspects and abstract away their differences. This suggests that it is not promising to
should not make strong assumptions about the internal structure of the
user (e.g., defining it through a particular language) as that would
only really be suitable for modeling a computer agent.  As a first
approximation we suppose that we want a representation
which, hiding all internal details, allows us to specify some
behavioral constraints for the user and then observe its interactions
with $S$ to prove its soundness with respect to the notion of erasure.

Following the approach we used with systems, a labelled transition
system seems to be suitable because it hides internal evolution and
shows precisely the possible sequences of interaction. But in this
case we need to consider a subtle aspect: the notion of erasure is
meaningful if we have a handle on the secrets which are to be erased,
since in particular we need to vary these secrets in order to observe
the effect of the variation.  A labelled transition system is too
abstract to make this easy. Instead we add just enough structure to
the labelled transition system to make it possible for us to control
the erasure related information.

%one-time values, which are changed in order to check if their presence is visible after erasure finishes or not, it turns out that a LTS for the user is too abstract to be useful for our purposes, because such changes are not easily modeled in its structure. So we have to enrich the LTS in a way that it is possible for us to control the erasure related information.

We therefore propose to represent the user through two components:
\begin{itemize}
\item a behavioral component $U$, the LTS, in which we include the computational aspect of the interaction with the system;
\item a memory component $\delta$, the store of secrets from which the $U$ component has to fetch all the erasure related information before sending them to the system.
\end{itemize}

Let us define these concepts precisely.
\begin{definition}[Abstract User]
An abstract user is given by an LTS $U=(\mathcal{S},\mathcal{L},\mathcal{T})$ such that:
\begin{itemize}
\item $\mathcal{S}=\{u_i:i$ is an index$\}$ is the set of (distinct) states, in which $u_0$ is the initial one.
\item $\mathcal{L}=\{a?v,a!v,a?BE,a?EE,i?v ~ | ~ i\in \mathcal{I}\}$, assuming $v \in \mathbb{V}$ like we did in the system description, is the set of labels for $U$, where the first four elements represent the counterpart of $S$ actions, while the last one represents the input request for a value coming from $\delta$. The idea is that the memory $\delta$ is a set of labelled values, where  $\mathcal{I}$ are the labels used to distinguish them. Since the only property we require on $\mathcal{I}$ is to be a set of unique indexes, we can simply consider $\mathbb{N}^{+}$ as a representative for $\mathcal{I}$.
 \item $\mathcal{T}\subseteq \mathcal{S} \times \mathcal{L} \times \mathcal{S}$ is the set of transitions.  We assume that also for users the input enabled condition holds, hence we have branching over all possible values of the domain in each input action of $U$.\end{itemize}
\end{definition}

It is worth pointing out that while $S$ is deterministic, $U$ does not need to be deterministic, and this makes our framework quite expressive. 
%Another particular aspect of user definition are $i?v$ actions, which basically represent the interface of the user to erasure related information.

In order to understand the intention of the $i?v$ actions, let us define the $\delta$ component: $\delta$ is simply a function $\mathbb{N}^{+} \rightarrow \mathbb{V}$ such that for each index it provides the value corresponding to that index.

% \begin{definizione}[Memory]
% With $\delta$ we denotate a function on $\mathbb{N}$ such that $\delta(i)=v$, where $v\in \mathbb{V}$. $\delta$ provides the input $v$ to the user $U$ whenever it emits an input request on channel $i$.
% \end{definizione}

%\draftnote{I think both the first and the second review had a problem with this pragraph, I guess removing it will make everything clearer: 
%It is useful to give some remarks about $\delta$: first, one could think about associating more then one value to each index, making $\delta$ a relation rather than a function. Although in principle it seems useful (i.e. one single index could represent a category of erasure related information, like a set of credit card numbers), we will see that for guaranteed erasure, erasure related items should be used only once in the interaction with the system, hence a relation-like $\delta$ is obtained splitting the set of values associated to one index in different function-like memories.
%SWITS
%\draftnote{Cut this paragraph?}
% With this premise, a memory $\delta$ is isomorphic to a stream.}
%, each of them being constant in giving always the same value for a certain index $i$. 
As for the system, we can define the notion of an instance $U(\delta)$, which is obtained by refining the behavior of an abstract user $U$ through a memory component $\delta$. In other words, the traces of $U(\delta)$ are also traces of $U$, but with the restriction that for
any transition labelled $i?v$ we have that $v = \delta(i)$.

Note that plugging together $U$ and $\delta$ is not a way to resolve the nondeterministic behavior, since users can be purely nondeterministic, independently of input operations. \footnote{We can remark here that it would also be rather natural to define $\delta$ as a transition system and to compose $U$ and $\delta$. But in any case the parallel composition would not be a standard CCS-style one since sometimes we will need to witness the values passed from $\delta$ to $U$.}

Now that we have a general structure for $U$, we can specify some
notational conventions about its behavior. The notion of trace is the same we gave for systems, hence let us consider $T(U)=T(u_0)$ as the set of traces for $U$, and $T(U(\delta)) \subseteq T(U)$. As for systems we introduce some convenient abbreviations for certain common sequences of interactions: let $u \ube{i}{w}u'$ be an abbreviation of $\exists u_2,u_3.
u \stackrel{a?BE}{\rightarrow} u_2 \stackrel{i?v}{\rightarrow}
u_3\stackrel{a!w}{\rightarrow} u'$,
%\draftnote{Is it the same for $u_3$?} 
and $u' \enErase{t}$ be a shorthand for $\exists u_1. u' \stackrel{t}{\twoheadrightarrow}u_1
\stackrel{a?EE}{\rightarrow}$. 

%The following example should clarify the reason it is enough to work with
%light traces.

%\begin{example}
%Suppose we would like to model a user $U$ who starts reading an erasure related input from $\delta$ in position $0$.\draftnote{streams started at 1. Problem?}
 %Assuming two possible values, 0 and 1, we can describe two different behaviors, namely setting up the erasure communication in the first case, and reading another secret for local purposes before doing that in the second. Then $U$ should be something like:
%\begin{itemize}
%\item $\mathcal{S}=\{u_0,u_1,u_2,u_3,u_4,u_5,u_6\}$;
%\item $\mathcal{T}=\{(u_0,0?0,u_1);(u_0,0?1,u_2);(u_2,1?v,u_3);(u_1,a?BE,u_4);(u_3;a?BE,u_5);(u_4,a!(0,0),u_6);(u_5$ $,a!(0,1),u_6)\}$
%\end{itemize}

%For $\delta$ we can consider a function like $\{0 \mapsto 1;1 \mapsto 40\}$ as well us $\{0 \mapsto 0; 1 \mapsto 23\}$, but nonetheless $U$'s visible behavior (two possible traces, $u_0 \stackrel{(a?BE)(a!(0,1))}{\twoheadrightarrow}u_6$ or $u_0 \stackrel{(a?BE)(a!(0,0))}{\twoheadrightarrow}u_6$) shows no difference but in the actual value of secret, hence we can simply hide the internal $i?v$ operations.
%\end{example}

With a clear definition for the user we can specify the communication model for $U$ and $S$ which will allow us to explore in detail the message exchanges through the channel $a$. For a given $U$ and $S$, $U|S$ denotes the transition system consisting of states formed from the product of the states of $U$ and $S$ respectively, with initial state $u_0|s_0$, and transitions given by 
% It is based on the parallel operator $|$, which is described through
the rules in Figure~\ref{tab:AbsCom}.
\begin{figure}[h]
\begin{center}
\begin{tabular}{cc}
$\inference*{u \stackrel{a?v}{\longrightarrow}u' & s \stackrel{a!v}{\longrightarrow}s'}{u|s\stackrel{v}{\longrightarrow}u'|s'}$
& 
$\inference*{u \stackrel{a!v}{\longrightarrow}u' & s \stackrel{a?v}{\longrightarrow}s'}{u|s\stackrel{v}{\longrightarrow}u'|s'}$
\\[3ex]
$\inference*{u \stackrel{a?BE}{\rightarrow} u' & s \stackrel{a!BE}{\rightarrow} s'}{u|s\stackrel{BE}{\longrightarrow}u'|s'}$
& 
$\inference*{u \stackrel{a?EE}{\longrightarrow}u' & s \stackrel{a!EE}{\longrightarrow}s'}{u|s\stackrel{EE}{\longrightarrow}u'|s'}$
\\[3ex]
$\inference*{ \ \ s \stackrel{b!v}{\longrightarrow}s'}{u|s\stackrel{b!v}{\longrightarrow}u|s'}$ $b \not = a$
& 
$\inference*{ \ \ u \stackrel{i?v}{\longrightarrow}u'}{u|s\stackrel{i?v}{\longrightarrow}u'|s}$ $i \not = a$
\end{tabular}
\caption{transition rules for $U | S$}\label{tab:AbsCom}
\end{center}  
\end{figure}
%\draftnote{The value is visible in the interaction. But those are what I understood as "transparent channels" / fine,D}

As we can see, the labels for the joint environment $U|S$ are different from 
both user and system, since they show the value of the message which is passing in a communication between user and system. But since we are modeling erasure using a noninterference-style, it is quite reasonable to have such a ``transparent'' channel between $U$ and $S$.

Following the same approach we adopted with $S$, there are some aspects of the user behavior we have to constrain in order to ensure its positive attitude towards erasure. Since some of them are related to its internal structure, we need the counterpart of system well formedness, which we define as follows.
\begin{definition}[Well Formed User]
A user $U$  is well formed iff for each trace $t \in T(U)$: 
\begin{itemize}
\item if $t=t'(a?BE)$, and $u_0 \stackrel{t}{\twoheadrightarrow} u$, then there exists an $i$ such that for all $v$, $u \stackrel{i?v}{\rightarrow} u_v \stackrel{a!v}{\rightarrow} u_v'$;
\item if $t=t_1(i?v)t_2$, $t_2 \neq \epsilon$,  
then $t_1$ ends in $a?BE$ and $t_2$ begins with $a!v$;
%\draftnote{Too strong? F: it seems to me ok: a?BE i?v and a!v are a triple which has to be strong, and the branching is on i?v}
\item  in each prefix $t'$ of $t $ the number of $\erusr{i}{v}$ labels is always greater or equal than the number of $a?EE$ labels, while if $u_0 \stackrel{t}{\twoheadrightarrow} u \not \rightarrow$ then the number of $\erusr{i}{v}$ in $t$ is equal to the number of $a?EE$ in $t$.
\end{itemize}
\end{definition}

While the last requirement is similar to the feature requested on system side, the first two state how the user actually exploits the memory saying that whenever a value is sent to the system at the beginning of an erasure exchange, it has been fetched from the memory. In order to explain why we need such property, let us consider the following example.
\begin{example}
As we have already said, the erasure property we are going to specify is described in a noninterference style, which means that we need a way to change secrets and let them flow to the system in a controlled way.
Suppose that we have such user $U$:
%\begin{itemize}
%\item $\mathcal{S}=\{u_0,...,u_5\}$
%\item $\mathcal{T}=\{(u_0,a?BE,u_1);(u_1;1?0;u_2);(u_1;1?1;u_3);(u_2,a!0,u_4);(u_3,a!0,u_4);(u_4,a?EE,u_5)\}$
%\end{itemize}
\begin{figure}[h]
\begin{center}
\includegraphics[width = 0.4\textwidth]{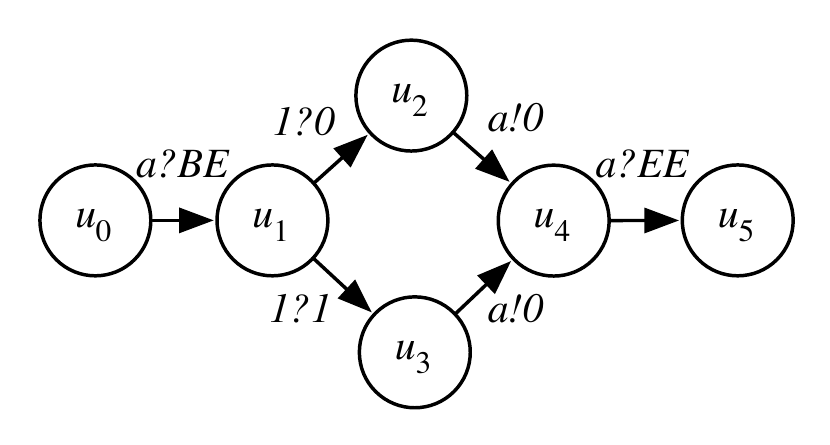}
\end{center}
%\caption{This is the caption at the bottom of the image}
\end{figure}

Although $U$ implements erasure protocol properly, it sends values for erasure independently of $\delta$, hence we cannot track the dependence between the value sent and subsequent behaviour by varying $\delta$.
% however we choose $\delta$ (assuming 0 and 1 as admissible values for position 1) $U$ is sending the same value to $S$, hence the memory is not a way to check the erasure property in a noninterference style. 
\end{example}

There are some other remarks on user well formedness it is worth pointing out: 
\begin{itemize}
\item we require that local read operations are only performed for the purpose of erasure related exchanges. Although this requirement is not strictly necessary to characterize an erasure friendly user, it makes the presentation of our results simpler, without restricting the range of user behaviors we are able to describe;
\item since we are defining the counterpart of well formed systems, the conventions related to well nested erasure blocks hold also for users.  
\end{itemize}

% we find a sequence $\erusr{i}{v}$ in an trace of $U$, we are sure that the value $v$ is coming from the $i$-th position of $\delta$.

Now we have a way to characterize a user which uses $\delta$ as a secret repository and applies the erasure protocol properly, but this is obviously not enough to ensure a \emph{responsible} use of such secrets. For sure we need to constrain the user in a way that a particular secret is used at most once during its interaction with $S$. This is what we dub \emph{secret singularity}.  
\begin{definition}[Secret Singularity]
% Let us suppose that $t\in T(U)$ and $\exists t_1$, $t_2$ and $\alpha=\erusr{i}{v}$ such that $t_1\alpha t_2=t$. Then $U$ is secret singular iff $\forall z_1$, $z_2$ such that $t=z_1\alpha z_2$ we have $t_1=z_1$ and $t_2=z_2$.
%SWITS
%\draftnote{check. Old definition allows $\ldots42!i\ldots 39!i\ldots$. Maybe that's not a problem but it looks odd. // F: sure, that was a stupid mistake, actually your definition is what I was looking for. But the "label" should be $i?v$ anyway, right?}
A user $U$ is secret singular if for any $i$ and any trace $t \in T(U)$, a label of the form $i?v$ (for any $v$) occurs at most once in $t$.  
\end{definition}
The following example should help to understand why we require secret singularity to user.
%SWITS 
%\draftnote{Maybe we can just recall the example in the introduction? F: I will prepare the picture, and then decide if it is better to use it or not after we know exactly how many pages we need}
\begin{example}\label{ex:a}
Let us consider the following system $S$, %in Figure \ref{sys1} 
which contains a chain of two non-nested erasure sessions and where $v$ and $w$ represent generic elements of $\mathbb{V}$ (hence one should imagine as many input edges as the size of that domain). 
%\begin{itemize}
%\item $\mathcal{S}=\{s_0,...,s_8\}$
%\item $\mathcal{T}=\{(s_0,a!BE,s_1);(s_1,a?v,s_2);(s_2,b!v,s_3);(s_3,a!EE,s_4);(s_4,a!BE,s_5);(s_5,a?v,s_6);(s_6,b!v,s_7);(s_7,a!EE,s_8)\}$
%\end{itemize}
\begin{figure}[h]
\begin{center}
\includegraphics[width = 0.6\textwidth]{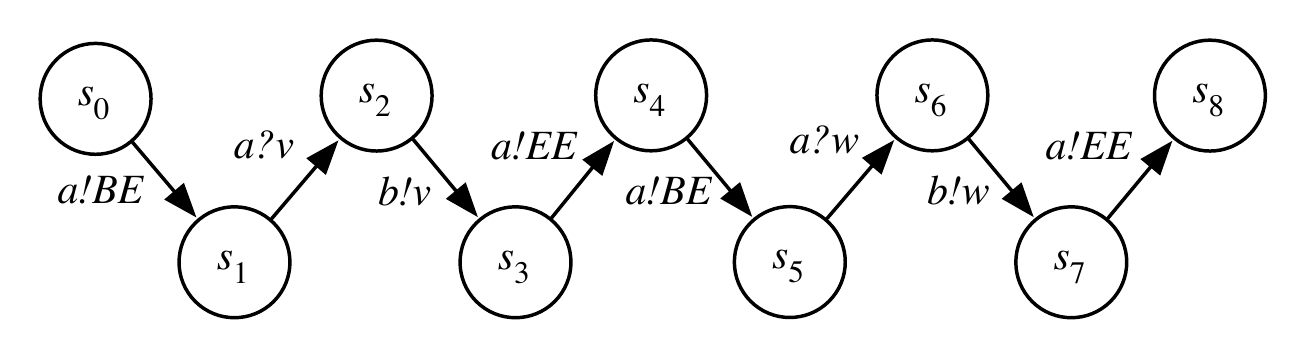}
\end{center}
%\caption{system for Example \ref{ex:a}}\label{sys1}
\end{figure}

This system obviously satisfies the input erasure definition of the previous section. But let us consider the user in Figure \ref{usr1}.
%\begin{itemize}
%\item $\mathcal{S}=\{u_0,...,u_8\}$
%\item $\mathcal{T}=\{(u_0,a?BE,u_1);(u_1,1?v,u_2);(u_2,a!v,u_3);(u_3,a?EE,u_4);(u_4,a?BE,u_5);(u_5,1?v,u_6);(u_6,a!v,u_7);(u_7,a?EE,u_8)\}$
%\end{itemize}
\begin{figure}
\begin{center}
\includegraphics[width = 0.6\textwidth]{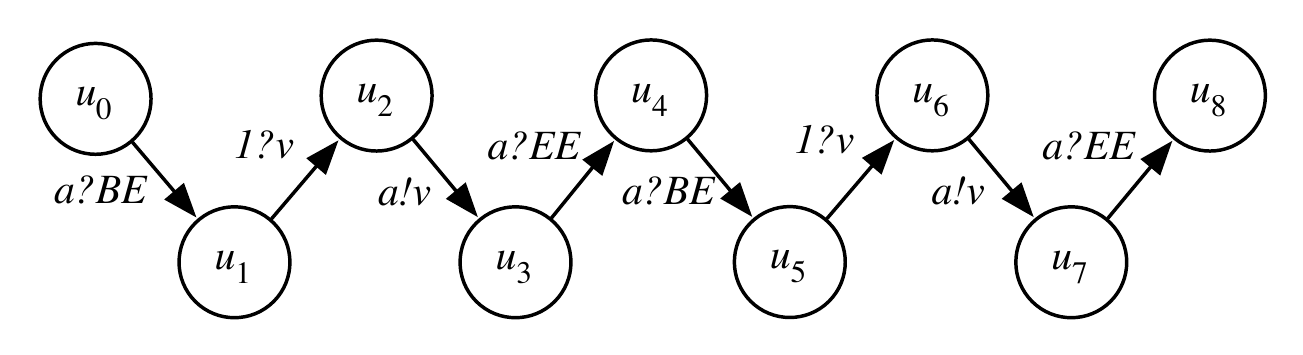}
\end{center}
\caption{user for Example \ref{ex:a}}\label{usr1}
\end{figure}
Suppose that the domain for the secret contains the values 0 and 1, then just two out of the four output operations admitted by $S$ are performed by $U$, namely $(a!0)(a!0)$ and $(a!1)(a!1)$. This happens because the secret is used twice, and we cannot prove that the behavior of $S$ and $U$ after the first erasure session is independent of secret value received (as we will see, this will be a requirement of the definition of composite erasure). Note that secret singularity is not satisfied by $U$, which is instead well formed.
\end{example}

%\draftnote{There were a lot of references to the notion of "erasure session", which (at this point) is clear for system but not defined yet for user. Hence I changed them in equivalent but more informal terms. If you do not like / think it does not work it is just enough to recover the old version}
In order to understand secret singularity, one should consider we are trying to build a structure in which each secret is completely independent of the others. Whenever the same portion of $\delta$ is used more than once, we immediately loose this independency, because changing that portion creates effects in more than one communication of the user to the system. It is therefore useful to think about the $i$ component of an $i?v$ request to $\delta$ as a timestamp representing the moment in which the user requests a secret: in this way secret singularity allows the \emph{actual value} of a secret being used twice, but requires that $\delta$ supports this duplication (i.e. contains two locations $i \not = j$ such that $\delta(i)=\delta(j)$).
 
Two other constraints on the user seem to be necessary to support the erasure mechanism properly:
\begin{itemize}
\item its behavior outside of an erasure phase
should not depend on values of secrets (what we call \emph{secret confinement}), and
\item such dependency should always preserve a stream-like behavior, which means that sequences of outputs proposed by the user have to be independent of both secrets and feedbacks received from the system, in order to ensure user will not be an unconscious storage for $S$.  We call this property \emph{stream ability}. 
\end{itemize}

Secret confinement is a property of the user behavior outside each erasure related part of its structure.
%UPDATED. IT WAS:  which does not contain erasure related exchanges. 
Stream ability, on the other hand, is specific to the part of the behavior within an erasure block. So, in order to describe those properties, we must define a way to separate those complementary parts of $U$'s traces, and then define their intended structure according to the informal requirements we have just stated.

%\draftnote{I would conclude the sentence in this way: "separate those complementary parts of $U$'s traces, and then define their intended structure according to the informal requirements we have just stated".}locate and constrain the various erasure related zones of $U$'s traces. %As a first  step toward this goal is defining them precisely.
%
% \begin{definition}[Erasure zone and erasure frontier]
% \draftnote{We need to discuss this definition}
% Let $U$ be a well formed user with the liveness property.
%  Let $u$ be a state in $\mathcal{S}$ such that $u \stackrel{a?BE}{\rightarrow} u_1$. For each $t \in T(u)$,
% \draftnote{$T(u)$ is a new notation}
%  which contains the closure of the erasure session started in $u$, we can define a prefix $t'$ such that $u \enErase{t'} u_v$, where $v$ represent the value fetched from $\delta$ by a $i?v$ operation in $t'$. The erasure zone centered in $u$, $E_\bullet(u)$, is then defined as the union of all such $t'$, and it comes together with a correspondent erasure frontier $E_\circ(u)$ which is the union of the $u_v$ states reached from $u$ through elements in $E_\bullet(u)$. 
% %We will refer to $u$ as the entrance of the erasure zone.
% \end{definition}

We focus first on the actions of a user which are not part of an erasure block. The following definitions facilitate this.

%\draftnote{I am sorry but I changed your definition trying to make the erasure zone definition easier. It seems to me I did not change the original meaning, if it is not I ask you sorry}
\begin{definition}[(Incomplete) User Erasure Session] %and Erasure Hiding]
A subtrace $t$ is an \emph{user erasure session} % centered in $u$} 
if $u_0 \twoheadrightarrow u \stackrel{t}{\twoheadrightarrow} u_v \twoheadrightarrow$ and $t=(\erusr{i}{v})(t_v \closee )$ such that it is balanced from the erasure point of view (so, for example, in $t_v$
 the number of $a?BE$ events is equal to the number of $a?EE$ events).

A \emph{user incomplete erasure session} %centered in $u$} 
is a subtrace $t$ such that $u_0\twoheadrightarrow u \stackrel{t}{\twoheadrightarrow}$ and $t=(\erusr{i}{v})t'$ where $(\erusr{i}{v})t''$ is not an erasure session for any  $t'' \preceq t'$.

%SWITS
%COMMENTED PART DUE TO UPDATE IN SECRET CONFINEMENT
%Given a trace $z$, we define the \emph{erasure-hidden trace} 
%$\erhide{z}$ to be the 
%trace obtained by removing all maximal erasure sessions and maximal incomplete erasure sessions from $z$.
%We extend $\erhide{\cdot}$ to sets of traces in a pointwise manner.
\end{definition}

Erasure sessions which share the same initial state are then grouped together in erasure zones.

%PART ADDED DUE TO UPDATE IN SECRET CONFINEMENT
\begin{definition}[Erasure Zone and Incomplete Erasure Zone] % which is now centered on secrets rather then on the state which sends it
The (incomplete) erasure zone of a user state $u$, $E_\bullet(u)$, is defined as 
\[
\begin{split}
 E_\bullet(u) &= \{ t   ~|~ u \stackrel{t}{\twoheadrightarrow}, \text{~$t$ is an (incomplete) erasure session} \}  \\
 % E_\circ(u) &=   \{ u_v ~|~ u \ube{i}{v}u_1\enErase{t} u_v, \text{~$t$ is balanced} \} 
\end{split}
\] 
\end{definition}

For each erasure zone we then define the frontier, namely the set of states which are immediately outside the erasure zone.

\begin{definition}[Erasure Frontier]
Let $E_\bullet(u)$ be an erasure zone. The union of $u_v$ such that $\exists u'. u\ube{i}{v}u'\enErase{t_v}u_v$ is called erasure frontier of $E_\bullet(u)$ and it is denoted by  $E_\circ(u)$.
\end{definition}

Now we have everything we need to state secret confinement:
\begin{definition}[Secret Confinement]
Let $U$ be a well formed user. $U$ satisfies secret confinement iff for each state $u$ we have that $\forall u_v,~u_w \in E_\circ(u)$,  $T(u_v)=T(u_w)$.
\end{definition}
%%OLD DEFINITION
%Let $U$ be a well formed user. $U$ satisfies secret confinement iff for each $\delta$ and $\delta'$ we have that $\erhide{T(U(\delta))} = \erhide{T(U(\delta'))}$.
 
% \begin{definition}[Secret confinement]
% Let $U$ be a well formed user with the liveness property, and $U^i$ its erasure free version. $U$ satisfies secret confinement iff for each $\delta$ and $\delta'$ we have that $U^i(\delta) \approx U^i(\delta')$.
% \end{definition}

%A user which satisfies the previous property cannot use the value of its secrets to drive its evolution outside erasure zones: intuitively this property allows the user to protect its own secrets by not remembering their value (not directly, nor indirectly). The following example shows how can a system fool a user which does not satisfies secret confinement.

%\DTNoteB{
It turns out that, for users satisfying secret confinement, differences among secrets are not reflected in differences among behaviors outside erasure zones, since each state in the frontier shows the same set of traces. We can therefore say that, for those users, secret values are \emph{confined} inside erasure zones. %} 
The following example shows how can a system fool a user which does not satisfies secret confinement.
\begin{example}
Let us consider the following user $U$:
%\begin{itemize}
%\item $\mathcal{S}=\{u_0,...,u_5\}$
%\item $\mathcal{T}=\{(u_0,a?BE,u_1);(u_1,1?v,u_2);(u_2,a!v,u_3);(u_3,a?EE,u_4);(u_4,a!v \% 10,u_5)\}$
%\end{itemize}
\begin{figure}[h]
\begin{center}
\includegraphics[width = 0.55\textwidth]{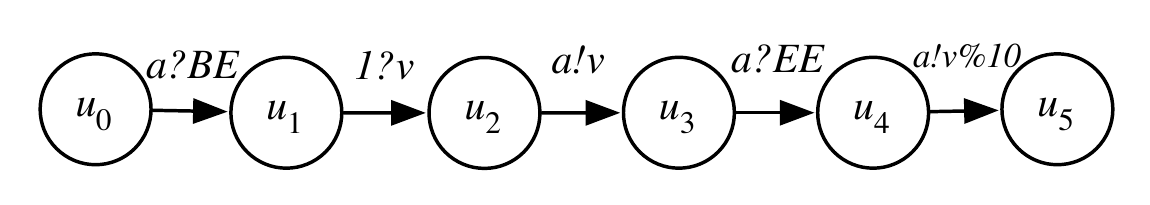}
\end{center}
%\caption{This is the caption at the bottom of the image}
\end{figure}

in which $v \%10$ represent the less significant digit of $v$. The final state of $U$ after the erasure session shows dependency on secret value, and it is quite easy then to figure out a system like:
%\begin{itemize}
%\item $\mathcal{S}=\{s_0,...,s_5\}$
%\item $\mathcal{T}=\{(s_0,a!BE,s_1);(s_1,a?v,s_2);(s_2,a!EE,s_3);(s_3,a?v,s_4);(s_4,b!v,s_5)\}$
%\end{itemize}

\begin{figure}[h]
\begin{center}
\includegraphics[width = 0.55\textwidth]{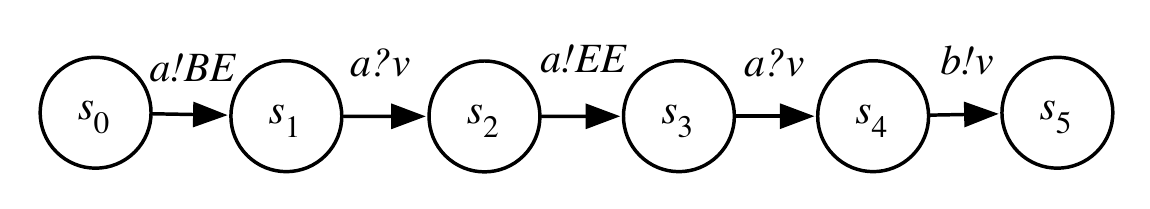}
\end{center}
%\caption{This is the caption at the bottom of the image}
\end{figure}
which is input erasing but it captures (and communicates) a portion of $U$ secret.
\end{example}

Now we need to consider the internal property of erasure zones. We have already said each member of the zone has to show a stream behavior, which basically consists in accepting any input from $S$ without being influenced by it. This is necessary because if $S$ is malicious,  $U$ can be tricked into storing erased data on behalf of $S$. 
\begin{example}\label{ex:streamab}
Let us consider the following system $S$:
%\begin{itemize}
%\item $\mathcal{S}=\{s_0,...,s_7\}$
%UPDATED
%\item $\mathcal{T}=\{(s_0,a!BE,s_1);(s_1,a?v,s_2);(s_2,a!f(v),s_3);(s_3,a!EE,s_4);(s_4,a?w,s_5);(s_5,b!f,s_6);(s_6,b!w,s_7)\}$
%\item $\mathcal{T}=\{(s_0,a!BE,s_1);(s_1,a?v,s_2);(s_2,a!f(v),s_3);(s_3,a!?w,s_4);(s_4,a!EE,s_5);(s_5,b!f,s_6);(s_6,b!w,s_7)\}$
%\end{itemize}
\begin{figure}[h]
\begin{center}
\includegraphics[width = 0.55\textwidth]{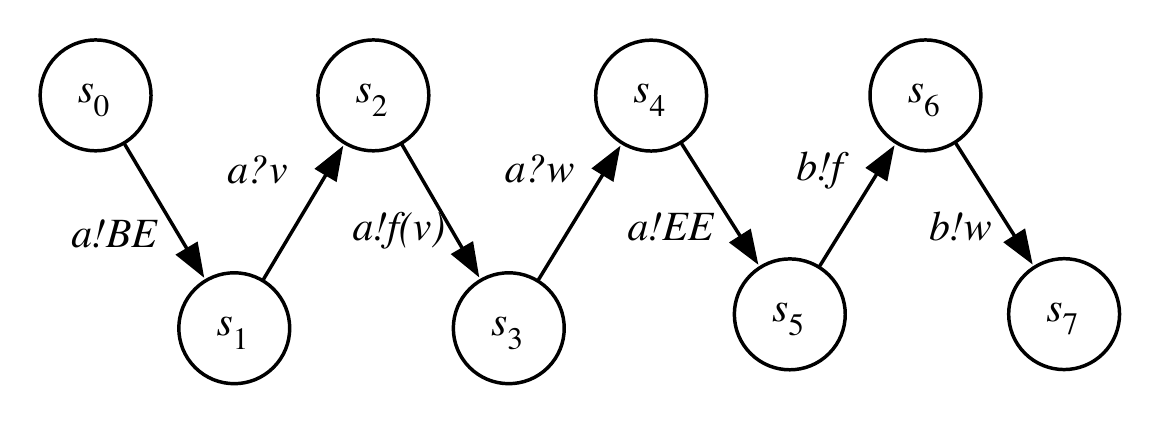}
\end{center}
%\caption{This is the caption at the bottom of the image}
\end{figure}

which represents %UPDATED 
a slight variation ofthe discount example (Section \ref{sec:intro}).
 
Let us also consider the user $U$ in Figure \ref{fig:usrstreab}.

%\begin{itemize}
%\item $\mathcal{S}=\{u_0,...,u_5\}$
%\item $\mathcal{T}=\{(u_0,a?BE,u_1);(u_1,1?v,u_2);(u_2,a!v,u_3);(u_3,a?w,u_4);(u_4,a!w,u_5);(u_5,a?EE,u_6)\}$
%\end{itemize}

\begin{figure}
\begin{center}
\includegraphics[width = 0.55\textwidth]{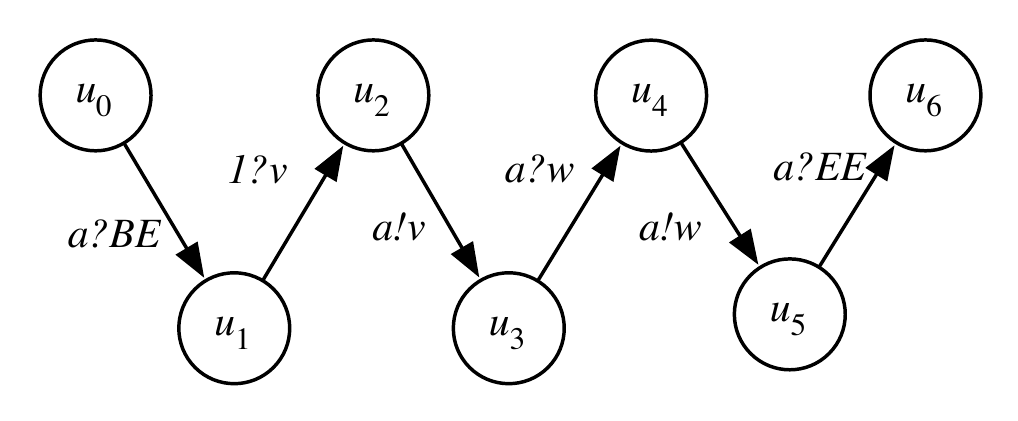}
\end{center}
\caption{system for Example \ref{ex:streamab}}\label{fig:usrstreab}
\end{figure}

Although $U$ satisfies secret confinement (after the reading operation the secret is sent and $U$ does not perform any particular action on it), $S$ is able to send the secret back to the user and wait $U$ to send it back again: in this way it will be possible for $S$ to send the encoding function $f$ and the value $f(v)$ to an attacker without breaking the input erasure property. 
\end{example}

%As we have already seen, inside $U$ there could be many different erasure sessions. Each erasure session always starts with an output of a secret. Due to internal nondeterminism, the same operation can be contained in many different traces of $U$. Moreover, due to input enabled condition we imposed on user, for each possible value in $\textbf{V}$ we have at least one trace which contains the corresponding erasure session for that value. But all those erasure sessions have in common the state from which the erasure protocol for a certain secret starts. We therefore need to collect each erasure session which starts from a given state, as it is specified in the following definition.

%ANTICIPATED 
%\begin{definition}[Erasure Zone] %[Alternative version?]
%The erasure zone of a user state $u$, $E_\bullet(u)$, is defined as 
%\[
%\begin{split}
 %E_\bullet(u) &= \{ t   ~|~ u \stackrel{t}{\twoheadrightarrow}, \text{~$t$ is an (incomplete) erasure session} \}  \\
 % E_\circ(u) &=   \{ u_v ~|~ u \ube{i}{v}u_1\enErase{t} u_v, \text{~$t$ is balanced} \} 
%\end{split}
%\] 
%\end{definition}

%\begin{definition}[Erasure zone]
%Let $U$ be a well defined user. Then we define $E_\bullet(u)$, the erasure zone centered in $u$, to be the set of all (incomplete) erasure session centered in $u$.
%\end{definition}

But how can we actually impose the stream behavior among elements in $E_\bullet(u)$? Since we are not interested in input actions, the idea is that we have to define a relation between traces in which the output values sent to the system depend only on predetermined decisions, which have to hold independently on $\delta$ content and system responses. We thus may find three possible cases while trying to compare two traces:
\begin{itemize}
\item if two non erasure outputs are performed, then the value has to be the same;
\item if two erasure outputs are performed, then the index of the secret used in the transaction has to be the same, because $\delta$ is not under $U$'s control;
\item if one erasure output has to be matched with a non erasure one, then the two values have to be the same.
\end{itemize}

The following output equality relation between traces states those requests formally.

\begin{definition}[Output Equality]
$\outeq$ is the largest symmetric binary relation defined over the set of rules contained in Figure \ref{tab:outeq}.
\begin{figure}[h!]
\begin{center}
\begin{tabular}{ccc}
$\inference*{t\outeq t'}{(\erusr{i}{v})t\outeq(a!v)t'}$ &  $\inference*{t\outeq t'}{(\erusr{i}{v})t\outeq(\erusr{i}{w})t'}$ & $\inference*{t\outeq t'}{(a!v)t\outeq(a!v)t'}$ \\[3ex]
\end{tabular}
\newline
\begin{tabular}{cc}
$\inference*{t\outeq t'}{lt\outeq t'}$ if $l\not \in \{ (\erusr{i}{v}),(a!v)\}$&  $\epsilon \outeq \epsilon$\\ 
\end{tabular}
\caption{Output Equality}\label{tab:outeq}
\end{center}  
\end{figure}

$t$ and $t'$ are output equivalent if $t\outeq t'$.
\end{definition}

Now we can state the stream ability for a user $U$.

\begin{definition}[Stream Ability]
A well formed user $U$ has the stream ability if for all states $u$ %which are erasure entrances 
it holds that $\forall~t,~t' \in E_\bullet(u)$, $t\outeq t'$.
\end{definition}

If a user satisfies all the definitions we stated so far, it is doing all its best to behave in a proper way with respect to erasure. We define such attitude as \emph{erasure friendliness}.

\begin{definition}[Erasure Friendly User $EF(U)$]
Let $U$ be a well formed user which satisfies secret singularity, secret confinement and stream ability. Then we define $U$ as an erasure friendly user, and we denote all its properties with the notation $EF(U)$.
\end{definition}

So far we worked mainly on the user structure, but we also need to rule out the possibility that $U$ is able to deadlock $S$ by performing the wrong communication, and we have to do that because $S$ is performing erasure only if it is able to reach the end of the erasure session. To avoid such kind of scenarios we impose \emph{liveness} on the user side, according to the following definition. 

\begin{definition}[Liveness]
Let $S=(\mathcal{S},\mathcal{L},\mathcal{T})$ be a well formed system. We say $s \in \mathcal{S}$ is a system input state if $s \stackrel{a?v}{\rightarrow} s_v \in \mathcal{T}$\footnote{Recall that input enabled condition on $S$  makes $s$ able to accept each possible input $w \in \mathbb{V}$.}. Let $U$ be a well formed user. Then $U$ has the liveness property for $S$ if:
\begin{itemize}
\item $u_0|s_0 \twoheadrightarrow u_1|s_1$ and $s_1$ is a system input state imply that $\exists~u_2$ such that $u_1 \stackrel{u}{\twoheadrightarrow} u_2 \stackrel{a!v}{\rightarrow}$ and $u_1|s_1 \twoheadrightarrow u_2|s_1 \rightarrow$;
\item $u_0|s_0 \twoheadrightarrow u_1|s_1$ and $s_1 \rightarrow s_2$, $s_1$ not a system input state, imply that $\exists~u_2$ such that $u_1 \twoheadrightarrow u_2$ and $u_1|s_1 \twoheadrightarrow u_2|s_2$. 
\end{itemize} 
\end{definition}

%%%%%%%%%%%%%
%%THIS IS A SECTION%
%%%%%%%%%%%%%
\section{Composite Erasure and Soundness of $U|S$}\label{sec:thm}
%\draftnote{``Joint erasure'' sounds more appropriate}

Since we are reasoning about $U|S$, let us import some notational conventions from the previous sections. If we assume that both $U$ and $S$ are well formed we are sure that each time a $BE$ communication is performed in a trace of $U|S$, an appropriate reading operation $i?v$ and a communication $v$ will follow immediately. So let  $u|s\cbe{i}{v}u'|s'$ be an abbreviation for $u|s\stackrel{BE}{\rightarrow}u_1|s_1\stackrel{i?v}{\rightarrow}u_2|s_1\stackrel{v}{\rightarrow}u'|s'$ and let $\enErase{t}$ denote a finite erasure session.

We can now state the definition for the erasure property which involves both system responsibility towards secret provided by user, and user's caution in treating potentially dangerous information from system in a safe way. 

%\draftnote{I suggest ``composite erasure''}
\begin{definition}[Composite Erasure $EC(U|S)$]
Let $S$ be a deterministic, well-formed system, with initial state $s_0$. Let $U$ be a well-formed user with the liveness property with respect to $S$ and let $\delta$ an arbitrary memory for $U$, such that together they define an instance $U(\delta)$. Consider then a trace $u_0|s_0 \stackrel{t}{\twoheadrightarrow} u_1|s_1 \cbe{i}{v}  u_2|s_2 \enErase{t'} u_3|s_3 \stackrel{z}{\twoheadrightarrow}$ of $U(\delta)|S$. 

Then $S$ and $U$ are composite erasing ($EC(U|S)$) if for any instance $U(\delta')$, such that $\delta'$ differs from $\delta$ only at position $i$, we have that $u_0|s_0\stackrel{t}{\twoheadrightarrow} u_1|s_1 \cbe{i}{w} u_2'|s_2' \enErase{t''} u_3'|s_3' \stackrel{z}{\twoheadrightarrow}$.
\end{definition}

The input erasure definition has two conditions to ensure that the scenario after erasure is independent of secrets. The first is related to the system behavior (equality between traces), while the second refers to the status of the streams associated to $S$, which have to be consumed in the same way. Here the second condition does not correspond to an explicit constraint, but since the behavior $z$ includes both contribution of $S$ and $U$, it is implicit in the equality of trace.   

%SWITS
%\draftnote{I find this paragraph hard to follow. What is the ``environment'' here? F: last version, needs to be checked}

Now the main result of this work can be stated, which relates local properties on both user and system sides to the notion of  composite erasure we have just defined.

\begin{theorem}[Soundness of $U|S$ wrt composite erasure]
Let $S$ be an input erasing system ($E(S)$) and $U$ be an erasure friendly user ($EF(U)$) which satisfies the liveness condition for $S$. Then $U|S$ satisfies the composite erasure properties $EC(U|S)$.
\end{theorem}

\begin{proof}

Consider a memory $\delta$ for $U$ and suppose that $U(\delta)|S$ produces a trace $u_0|s_0 \stackrel{t}{\twoheadrightarrow} u_1|s_1 \cbe{i}{v}  u_2|s_2 \enErase{u} u_3|s_3 \stackrel{z}{\twoheadrightarrow}$. 

Now consider a $\delta'$ differing from $\delta$ only at position $i$, where
$\delta'(i)=w$.
We then need to prove $U(\delta')$ together with $S$ produces the trace $u_0|s_0\stackrel{t}{\twoheadrightarrow} u_1|s_1 \cbe{i}{w} u_2'|s_2' \enErase{u'} u_3'|s_3' \stackrel{z'}{\twoheadrightarrow}$ such that $z=z'$.

Let us start showing that $U(\delta')|S$ produces $u_0|s_0\stackrel{t'}{\twoheadrightarrow} u_1'|s_1'$ such that $t=t'$, $s_1=s_1'$ and $u_1=u_1'$. Applying user well formedness and secret singularity we have that $t$ is a trace which does not contain local reads on index $i$. Hence also $U(\delta')|S$ behaves according to $t$, reaching the same intermediate states $U(\delta) | S $ was able to reach, thus we have $t=t'$, $s_1'=s_1$ and $u_1'=u_1$. 

Now we need to show $u_1|s_1 \cbe{i}{w} u_2'|s_2' \enErase{u'} u_3'|s_3'$ when the system is communicating with the user instance $U(\delta')$. Due to $S$ determinism, $s_1$ emits a request of a secret value as it happens in the $U(\delta)|S$ execution. Then, since $U$ is input enabled an well formed, it will fetch the secret contained in $i$-th location of $\delta'$ (namely $w$) and send it to the system. This is represented by the first step $u_1|s_1 \cbe{i}{w} u_2'|s_2'$.

According to the definition of erasure zone, both $\erusr{i}{v}$ and $\erusr{i}{w}$ are the first actions of two erasure sessions contained in the same erasure zone $E_\bullet(u_1)$, hence stream ability applies. This, together with user liveness, implies that it exists a state $u_3'$ such that $u_2'|s_2'  \enErase{u'} u_3'|s_3'$ and the sequence of output contained in $u'$ and $u$ are the same. This property, which holds for an arbitrary value related to secret $i$, allows us to apply input erasure condition on $S$ side, having $s_3 \stackrel{\alpha}{\twoheadrightarrow} $ and $s_3' \stackrel{\alpha}{\twoheadrightarrow}$. 

Since secret confinement holds, we have that each trace of $u_3$ finds a correspondent trace on $T(u_3')$, hence due to $S$ determinism we have $u_3'|s_3' \stackrel{z}{\twoheadrightarrow}$ as required.
\end{proof}

%%PART REMOVED DUE TO CHANGES ON SECRET CONFINEMENT
%Now we have to distinguish two cases. If $E\_bullet(u_1)$ was maximal, then it is straightforwardly possible to apply secret confinement on user side, which gives equality on traces going through $u_3$ as well as $u_3'$. Then for the liveness property on the user side we have $u_3'|s_3' \stackrel{z}{\twoheadrightarrow}$ and $u_3|s_3 \stackrel{z}{\twoheadrightarrow}$. Otherwise $E\_bullet(u_1)$ was a portion of a maximal erasure session $E\_bullet(u')$. Then, due to system determinism, it is possible to split the $z$ trace in two contributions $z_1$ and $z_2$ such that $u_3|s_3 \stackrel{z_1}{\twoheadrightarrow}u_4|s_4$ and  $u_3'|s_3' \stackrel{z_1}{\rightarrow}u_4'|s_4'$ hold for stream ability of $E\_bullet(u')$, while $u_4|s_4 \stackrel{z_2}{\twoheadrightarrow}$ and $u_4'|s_4' \stackrel{z_2}{\twoheadrightarrow}u_4|s_4$ hold for secret confinement applied on $E\_bullet(u')$.

%\end{proof}

\section{Related Work}\label{sec:related}

%SWITS
%\draftnote{is it work or works? If it is work a sentence in the outline has to be updated}

This work has studied an interactive form of block structured
erasure. More expressive non-block structured erasure notions have
been studied by Chong and Myers \cite{er1,Chong:Myers:End}, where the
erasure point is associated with an arbitrary condition on the
computation history of the system, but these do not consider an
interactive system. Erasure and declassification for multithreaded
programs is studied in \cite{Jiang+:Handling}, although what the
authors call erasure is merely the much weaker notion of upgrading the
security classification of a variable.

Erasure as defined here is a close relative of
noninterference. Noninterference for interactive systems have been
studied extensively - for example in language based setting
\cite{str1} enforce language based security conditions in an
imperative language which allows input-output interactions, as well as
nondeterminism. This result is obtained through the notion of
\emph{strategy}, which is used to represent the behavior of the user
(secrets provider) as a function of previously exchanged values. In
\cite{str2} the results of \cite{str1} are refined by classifying the
strategies according to their expressive power and proving that it is
sufficient to model the user strategy as a simple stream of values in
the case that the system is deterministic.  With respect to these
works our modeling goal is rather different. A user, for example, is
modeled as having a stream-like behavior not because it is
\emph{sufficient} (as in the noninterference case \cite{str2}) but because it is
\emph{necessary} in order that the user does not feed data
into a system which depends on previously erased values.
%   is different because we are trying to obtain an abstract definition of the data provider which has to be, in the same time, much more liberal than the notion of strategy, but simple and structured enough to keep (a slight variation of) the language in \cite{Hunt:Sands:ESOP08} sound with respect to erasure.

Erasure properties are related to \emph{usage control} policies (e.g., 
\cite{Park:Sandhu:Towards,Park:Sandhu:UCONABC,Pretschner+:Distributed})
%\emph{obligations} \cite{er2}
-- however the fundamental difference appears to be that the usage control area is related to access control rather than information flow control, so for example the scenario in the introduction where the user unwittingly stores a secret on behalf of the system would seem to be out of reach of those methods.

Another area where connections to erasure may be seen is in a recent account of
\emph{opacity} at the level of transition systems \cite{Bryans+:Opacity}, although at the time of writing deeper connections between the concepts remain to be established.

\section{Conclusions and Further Work}\label{sec:conclusions}
%In \cite{Hunt:Sands:ESOP08} Hunt and Sands studied a notion of \emph{information erasure} for systems which receive secrets intended for limited-time use. 

The definition of  \emph{information erasure} studied in \cite{Hunt:Sands:ESOP08}
considers an interactive environment, composed by an erasing system (data processor) and many users (secrets providers), but did not define the latter precisely. Here we tried to provide a formal but abstract model of the user and a collection of requirements called \emph{erasure friendliness} on its behavior, which are sufficient to ensure \emph{composite erasure}, namely the erasure property obtained by composing an erasure friendly user together with an erasing system. 

In doing so we identify stronger requirements on the user than those
informally described in the previous work, which are basically a
countermeasure for the liberal interactions an erasing system can
perform inside an erasure session. In particular it turns out that
erasure related data, as well as system feedbacks, should not be
examined by the user during an erasure session (i.e. communicated or
received without understanding their actual value).

The results we achieved in this work suggest a wide range of future developments. 

\paragraph{Completeness}
Considering our contribution in isolation, a natural enhancement is trying to prove the completeness on our approach: % ; in particular, the erasure friendly definition, although reasonable, is not formally justified and since it is rather onerous from the user perspective, 
we would like to show that erasure friendliness is not only a sufficient condition but a necessary one. In other words if a user is not erasure friendly then we would like to show the existence of an input erasing system for which composition yields system which is not jointly erasing.

% the only way to enforce the erasure mechanism properly.

\paragraph{Multilevel and Mutual Erasure}
% Apart from that, there are a lot of tasks to close the gap between the previous work on erasure and our achievements. As we have already said in the introduction, a
% Although we are trying to capture an abstract notion of erasure, we proposed a  of simplifications with respect to the approach adopted in \cite{Hunt:Sands:ESOP08}. It is therefore necessary to understand how our results fit in the old scenario. 
One of the simplifications we made to the earlier definitions of erasure was to treat only ``complete erasure''. In the general scenario we have a multilevel lattice and erasure from one level to a higher level. 
Adding the multilevel security structure on the system side, in which each security level comes equipped with a communication channel, allows to recover the original notion of ``erasure below a certain level $b$''. But to have this it is necessary to rearrange the system representation, in order to examine contents of output actions of $S$ below the intended erasure level. More difficult is to recover the system full input-output behavior: according to our first intuitions on the problem, allowing the system to perform I/O operations in each channel requires the classification of each principals of the communication according to three possible roles (secret provider, secret observer -- a user higher than erasure level, and secret co-manager -- a user who receives information tainted by other users' secrets, which are supposed to be erased). Such complexity increases the number of possible unwanted leaks of secrets, and seems to require a further strengthening of behavioral constraints for system users.

%\paragraph{Tradeoffs between User and System}
%The definition of system erasure which we have studied is not set in stone. 
%
%  Other theoretical issues are worth to be explored. The first is related to the behavioral condition we have in the system side to guarantee erasure: since we require that the only invariant aspect of its behavior is related to the overall number of input requests inside an erasure session, we had to impose the stream ability on the user side, which reduces dramatically the expressivity of the user inside an erasure related interaction with $S$. 
%It seems promising to explore strengthening of the erasure condition for the system in order to allow a more liberal user to be judged as erasure friendly.
%In particular in the definition of an input erasing system, the constraint on the communications during an erasure session are very weak -- simply requiring that the number of input actions does not depend on the secret to be erased. This in turn forces the user to be very conservative. 

\paragraph{Programming with Erasure}
The definition of erasure that we have studied is, like
noninterference, a very strong property. In practice -- as with
noninterference -- this means that one may have to
program systems in a particular style in order to satisfy erasure, or 
to introduce controlled mechanisms to weaken the erasure requirements, analogous to 
\emph{declassification}. As an example, consider the case of a credit card
transaction where the card can be denied by the bank. It would be
natural for the system (merchant) to request a new card from the user. But if
this were done in the same session then this would not be erasing: the
number of inputs from the user would depend on the first card number. This can be solved in two ways (i) by structuring the system so that the first interaction is concluded as a failed interaction, and inviting the user to engage in a subsequent purchase (where hopefully the user will try a valid card), or (ii) by \emph{revoking} the erasure promise on discovering an invalid card. It would seem that the various \emph{dimensions of declassification} \cite{Sabelfeld:Sands:Dimensions} would be suitable to study erasure revocation mechanisms.

\paragraph{Implementation: A Safe Wallet}
There are also implementation perspectives related to our results:
since we represented the user through two distinct components ($U$ for
the computation, $\delta$ as a secret repository), we have begun
looking at how it would be possible to define an advanced memory
component (a ``wallet'') $W$ which acts as a proxy between a user and
a system. The idea would be to achieve an erasure-friendly user
(modulo liveness issues) as a composition $A|W(\delta)$ of an
\emph{arbitrary} user $A$ and a generic safe wallet $W$ containing secrets
$\delta$. Since it is reasonable to consider the
    wallet as a real implemented agent, we believe it will be an
    effective mechanism to enforce our behavioral properties in a
    purely syntactic manner. The wallet $W$, together with a certified
  input erasing system $S$ (a trusted authority mechanism could certify
  that property, using the type system proposed in
  \cite{Hunt:Sands:ESOP08} for example), would therefore create a
  composite erasing environment without requiring the real user to
  constrain his own behavior.

\paragraph{Acknowledgements}
Sebastian Hunt provided valuable inputs throughout the development of this paper. Many thanks to our colleagues in the ProSec group at Chalmers and to the anonymous referees for useful comments and observations.
This work was partially financed by grants from the Swedish research agencies VR and SSF, and the
European Commission IST-2005-015905 MOBIUS project.

% which is able to lift an arbitrary user to the level of an erasure friendly one. To obtain this result the idea is to add some computational features to $\delta$, in order to let it be expressive enough to protect the user behavior inside erasure session, modifying its interface dynamically to avoid harmful I/O operations.

\bibliographystyle{eptcs} % or whatever you prefer
\bibliography{bibliography}%,literature}
\end{document}